\documentclass[journal]{IEEEtran}
\usepackage{amsmath,amssymb,paralist,subfigure,graphicx,color,stmaryrd,mathrsfs,url,algorithm2e,soul}

\usepackage[table]{xcolor}
\usepackage{cuted,mathtools,lipsum,cite}
\newtheorem{lem}{Lemma}
\newtheorem{ass}{Assumption}
\newtheorem{thm}{Theorem}

\newtheorem{rem}{Remark}

\newtheorem{cor}{Corollary}

\def\mb{\mathbf}

\def\mc{\mathcal}

\DeclareMathOperator*{\argmax}{argmax}

\begin{document}
\title{\huge \bf  Consensus-Based Distributed Estimation in the presence of Heterogeneous, Time-Invariant Delays
}
\author{Mohammadreza  Doostmohammadian, \IEEEmembership{Member, IEEE},    Usman~A.~Khan, \IEEEmembership{Senior Member, IEEE}, Mohammad~Pirani, \IEEEmembership{Member, IEEE}, and Themistoklis Charalambous, \IEEEmembership{Senior Member, IEEE}	
\thanks{M.  Doostmohammadian and T. Charalambous are with the School of Electrical Engineering at Aalto University, Finland (name.surname@aalto.fi). M.  Doostmohammadian is also with the Faculty of Mechanical Engineering at Semnan University, Iran (doost@semnan.ac.ir).}
\thanks{{U. A. Khan} is with the Electrical and Computer Engineering Department at Tufts University, MA, USA (khan@ece.tufts.edu).} 
\thanks{{M. Pirani} is with the Department of Electrical and Computer Engineering, University of Waterloo, Canada (mpirani@uwaterloo.ca).
}}
\maketitle
\thispagestyle{empty} 
	
\begin{abstract}
	Classical distributed estimation scenarios typically assume timely and reliable exchanges of information over the sensor network. 
	This paper, in contrast, considers single time-scale distributed estimation via a sensor network subject to transmission time-delays. The proposed discrete-time networked estimator consists of two steps: (i) consensus on (delayed) a-priori estimates, and (ii) measurement update.  The sensors only share their a-priori estimates with their  out-neighbors  over (possibly) time-delayed transmission links. The delays are assumed to be fixed over time, heterogeneous, and known. We assume distributed observability instead of local observability, which significantly reduces the communication/sensing loads on sensors. Using the notions of augmented matrices and Kronecker product, the convergence of the proposed estimator over strongly-connected networks is proved for a specific upper-bound on the time-delay. 

\keywords  Distributed estimation, consensus, Kronecker product, communication time-delays 	
\end{abstract}


\section{Introduction} \label{sec_intro}
\IEEEPARstart{L}{atency} in data transmission networks may significantly affect the performance of decision-making over sensor networks and multi-agent systems \cite{HOFBAUER}. In particular, time-delays may cause instability in networked control systems which are originally stable in the corresponding delay-free case. For example, the consequence of communication delays on the consensus stability  are discussed in \cite{Themis_delay,sharifi2020finite,ghaedsharaf2021centrality} among others, 
and centralized observer design 
are  discussed in \cite{sundaram2007delayed,schenato2006optimal}. This work extends to \textit{distributed} estimation over a sensor network with random communication time-delays.

The literature on distributed estimation spans from multi time-scale scenarios to single time-scale methods.  
The former case  requires many iterations of averaging/data-sharing (consensus/communication time-scale) between two consecutive system time-steps (system time-scale) \cite{he2019secure,battilotti2021stability}, where the estimation performance tightly depends on the number of consensus iterations. This  is less efficient in terms of computational and communication loads on sensors and, further, requires much faster data sharing/processing rate which might be inaccessible over large networks. In terms of observability, in the multi time-scale method, number of communication/consensus iterations is greater than the network diameter, and therefore, all sensors eventually gain all state information (and system observability)  between every two system time-steps. In the single time-scale, however, every sensor performs only one iteration of consensus, and therefore, many works require the system to be \textit{locally} observable in the neighborhood of the sensors \cite{kar2014distributed,sayed-kf,das2016consensus,mohammadi2015distributed,mo2020distributed,jenabzadeh2020distributed,zou2019moving,liu2017distributed}; in this work, we assume \textit{global} observability as in \cite{jstsp,silm2020simple,mitra2018distributed,icassp13}. Recall that local observability mandates: (i) more network connectivity, and/or (ii)  access to more system outputs at each sensor, and may considerably increase the communication/sensing-related costs \cite{pequito2014optimal,spl18,lcss2020}. This work, however, considers least connectivity requirement   (strong-connectivity) and least outputs at each sensor (one output), while addressing transmission delays.

The networked estimator in this paper is single time-scale, where sensors perform one consensus iteration on (possibly) \textit{delayed} a-priori estimates in their in-neighborhood, and then,  measurement-update using their own outputs.
As in\cite{Themis_delay}, we consider  arbitrary time-delays  at every communication link, but the delays are time-invariant and known. The delays are bounded so that no information is lost over the network and the data would eventually reach the recipient sensor.  
To avoid considering a trivial case, this work makes no assumption on the stability of the  linear system. Further, similar to \cite{sayed-kf,mohammadi2015distributed}, we assume that the system is full-rank as in structurally-cyclic/self-damped systems \cite{chapman2013strong,tnse_selfdamp}. We adopt the notions of augmented representation \cite{Themis_delay} and the Kronecker network product \cite{tsipn_kron} to simplify the convergence analysis. We show that feedback gain design in the absence of delays via the Linear-Matrix-Inequality (LMI) in \cite{jstsp,usman_cdc:11} also results in  stable  estimation  for  some upper-bounded delayed cases. Further, we provide a solution to design delay-tolerant networked estimators for a given bound on the delays. Therefore, the gain design requires no information other than the bound  on the delays and, the LMI complexity is determined by the original low-order system and not the high-order augmented one. Note that in this work, we assume \textit{no measurement delays}; this is because the sensors take direct measurements, while spatially distributed in large-scale with (possibly) delayed communications. Further, as proved in \cite{schenato2006optimal}, stability  depends \textit{only} on the measurement \textit{packet loss}, not the \textit{packet delays}.

\textit{Paper organization:} Section~\ref{sec_fram} provides some preliminaries and problem statement. Section~\ref{sec_pro} states our main results on delay-tolerant distributed estimation. Section~\ref{sec_sim} provides the simulations, and Section~\ref{sec_conc} concludes the paper.


\section{The Framework} \label{sec_fram}
\subsection{System-Output Model} \label{sec_sys} 
We consider discrete-time system and measurements as,
\begin{align}\label{eq_sys}
	\mb{x}_{k} &= A \mb{x}_{k-1} + \mb{\nu}_{k},~k\geq 1\\
	\mb{y}_{k} &= C \mb{x}_{k} + \mb{\zeta}_k, \label{eq_y}
\end{align}
with $\mb{x}_k \in \mathbb{R}^n$ as system states, $\mb{y}_k \in \mathbb{R}^N$ as system outputs, $\mb{\zeta}_k\sim\mc{N}(0,R)$ and $ \mb{\nu}_k \sim \mc{N}(0,Q)$ as independent  noise variables, all at time-step $k$. It is \textit{not} assumed that $\rho(A)<1$ (potentially unstable system), while $\det (A)\neq 0$ (full-rank), 
 where $\rho(\cdot)$ and $\det(\cdot)$  are the spectral radius and determinant, respectively. Examples of such full-rank systems are structurally cyclic \cite{tnse_selfdamp} and \textit{self-damped}\footnote{A linear system  is self-damped if its matrix $A$ has non-zero diagonal entries \cite{tnse_selfdamp,chapman2013strong}. Similarly, network $\mc{G}=\{\mc{V},\mc{E}\}$ is self-damped if for every node $i \in \mc{V}$ we have $(i,i) \in \mc{E}$, i.e., there is a self-link at every node.} systems \cite{tnse_selfdamp,chapman2013strong} among others, prevalent in, e.g., social opinion dynamics \cite{ghaderi2014opinion}.  
In this work, without loss of generality, we assume $N$ sensors each with one output ${y}_{k}^i$. 

\begin{ass} \label{ass_AC}
	Every sensor knows the system matrix $A$. The pair $(A,C)$ is observable (similar arguments for detectable case), implying \textit{global} observability. However, in general, the pair $(A,  C_i)$ is not necessarily observable at any sensor $i$. 
\end{ass}

Note that the output matrix $C$ can be defined via the graph-theoretic methods  to ensure (structural)  $(A,C)$-observability. From \cite{jstsp,pequito2015framework,lcss2020}, having one output from (at least) one state associated with every irreducible block of the adjacency matrix $A$ ensures structural observability. The optimal output selection strategies are also of interest as in \cite{ru2010sensor,tnse_selfdamp,pirani2021strategic}. 

\subsection{Preliminaries on Consensus Algorithms} \label{sec_pre}
Consider discrete-time consensus algorithms over a network of sensors $\mc{G}=\{\mc{V},\mc{E}\}$ ($\mc{V}$ as the node set and $\mc{E}$ as the link set) with $\mb{z}_{k}$ as the state of  sensors at time $k$, which evolves as $\mb{z}_{k} = P \mb{z}_{k-1}$. Matrix $P$ (as the consensus weight) represents the communication between the sensors via graph $\mathcal{G}$. 
The sensor network is in general directed. For notation simplicity denote $P(i,j)$ by $p_{ij}$, where $0<p_{ij}<1$ if $(j,i) \in \mc{E}$ and $0$ otherwise. $P$  is \textit{row-stochastic}, i.e., $\sum_{j=1}^N p_{ij} = 1$, and $p_{ii} \neq 0$ for all $i$. Further,  the network $\mc{G}$ needs to be (at least) \textit{strongly-connected} (SC), i.e., there is a path from every node $i$ to every node $j$ in $\mathcal{V}$, implying  that fusion matrix $P$ is \textit{irreducible}, also called stochastic, indecomposable, and aperiodic (SIA) and  $\lim_{k \rightarrow \infty} P^k = \mb{d}\mb{1}_N^\top 
$ \cite{Themis_delay}, with $\mathbf{1}_N$ as all-ones vector of size $N$. For such SIA matrix,  $\rho(P)=1$. 
 
\subsection{Delay Model} \label{sec_delay}
In this work, it is assumed that the data-transmission over the link $(j,i)$ from sensor $j$ to sensor $i$  has a-priori unknown bounded (integer)  time-delay, $\tau_{ij}$, where $0\leq \tau_{ij} \leq \overline{\tau} <\infty$, and $\overline{\tau}$ is an upper bound to the delays in all links. The messages are \textit{time-stamped}, so the recipient knows the time-step the data was sent.
Further, $\tau_{ii}=0$, i.e., every sensor $i$ knows its own state with no delay. To model the delayed state vectors we adopt the notations in \cite{Themis_delay}. In a network of $N$ sensors, define an \textit{augmented state} vector $\underline{\mb{x}}_k = \left(\mb{x}_k; \mb{x}_{k-1}; \dots; \mb{x}_{k-\overline{\tau}} \right)$  with '$;$' as column concatenation, and $\mb{x}_{k-r} = \left({x}^1_{k-r};  \dots; {x}^n_{k-r} \right)$ for $0\leq r \leq \overline{\tau}$. Then, for a given $N$-by-$N$ matrix $P$ and maximum delay $\overline{\tau}$, define the \textit{augmented matrix} $\overline{P}$ as,
 \begin{align} \label{eq_aug_W} \small
	\overline{P}= \left( 
	\begin{array}{cccccc}
		P_0 & P_1 & P_{2} & \hdots & P_{\overline{\tau}-1} & P_{\overline{\tau}} \\
		I_N &   0_N & 0_N &\hdots  & 0_N& 0_N\\
		0_N & I_N & 0_N &  \hdots  & 0_N & 0_N  \\
		0_N &  0_N & I_N  &  \hdots  & 0_N & 0_N  \\
		\vdots & \vdots & \vdots & \ddots & \vdots & \vdots \\
		0_N & 0_N & 0_N & \hdots & I_N & 0_N
	\end{array}	
	\right), \normalsize
\end{align} 
with $I_{N}$ and $0_{N}$ as the identity and zero matrix of size $N$. The non-negative matrices $P_r$ are defined based on the time delay $0\leq r \leq \overline{\tau} $ on the network links as follows,
\begin{align}
	P_r(i,j) = \left\{
	\begin{array}{ll}
		p_{ij}, & \text{If}~ \tau_{ij}=r  \\
		0, & \text{Otherwise}.
	\end{array}\right.
\end{align}
Assuming fixed delays, for any $(j,i) \in \mc{E}$ in the given communication network $\mc{G}$, \textit{only one of $P_0(i,j),P_1(i,j), \hdots, P_{\overline{\tau}}(i,j)$ is  equal to  $p_{ij}$  and the rest are zero.} This implies that the row-sum of each of the first $N$ rows of $\overline{P}$ and $P$ are equal, i.e., $\sum_{j=1}^{N(\overline{\tau}+1)} \overline{P}(i,j) =  \sum_{j=1}^{N} {P}(i,j)$ for $1\leq i \leq N$ and $P = \sum_{r=0}^{\overline{\tau}} {P}_r$ for $k \geq 0$. \textit{Therefore, in case of having a row-stochastic matrix $P$, the augmented matrix $\overline{P}$ is also row-stochastic\footnote{Note the subtle difference between our notation vs. Ref. \cite{Themis_delay}. In \cite{Themis_delay} \textit{column} augmented matrix is introduced, while we consider \textit{row} augmented matrices.}.} In the proposed estimator, we do not need the matrix $\overline{P}$ and it is only defined to simplify the mathematical analysis. Similar to \cite{Themis_delay}, the knowledge of the probability distribution of $\tau_{ij}$s is not needed for the analysis in this paper.  
\begin{ass} \label{ass_tau}
	 For the time-delay $\tau_{ij}$ on  link $(j,i)$:
	 \begin{enumerate} [i)]
	 	\item 
	 	The delay is known and bounded $\tau_{ij} \leq \overline{\tau}$. 
	 	The upperbound $\overline{\tau}$ guarantees no lost information, i.e., the data sent from sensor $j$ at time
	 	$k$ would eventually reach the recipient sensor $i$ (at most) at time $k+\overline{\tau}$ (a $\overline{\tau}+1$-slot transmission buffer).
	 	\item Delay $\tau_{ij}$ is arbitrary, fixed, and  may/may-not differ for  $(j,i)$ links (heterogeneous/homogeneous delays). 
	 \end{enumerate}
\end{ass}

\subsection{Problem Statement}
The problem in this work is to design a networked estimator for the system-output model \eqref{eq_sys}-\eqref{eq_y} satisfying Assumption \ref{ass_AC},
where every sensor relies only on its partial system output (partial observability) and the received  (possibly delayed satisfying Assumption~\ref{ass_tau}) information from its in-neighbors. This work particularly differs from  \cite{kar2014distributed,sayed-kf,das2016consensus,mohammadi2015distributed,mo2020distributed,jenabzadeh2020distributed,zou2019moving,liu2017distributed} via the following remark.    
\begin{rem} \label{rem_ACj}
Let $\mc{N}_i=\{j|(j,i) \in \mc{E}\}$ denote the set of in-neighborhood of  sensor $i$ over the network $\mc{G}$. The pair $(A, \sum_{j \in \mc{N}_i} C_j)$ is not necessarily observable at any sensor $i$, implying no \textit{local} observability assumption. 
\end{rem}

\section{Distributed Estimation in Presence of Delays} \label{sec_pro}
Every sensor $i$ performs the following two steps for distributed state estimation in the presence of time-delays, 
\small
\begin{align}\label{eq_p} 
	\widehat{\mb{x}}^i_{k|k-1} =& p_{ii}A\widehat{\mb{x}}^i_{k-1|k-1} + \sum_{j\in\mathcal{N}_i} \sum_{r=0}^{\overline{\tau}} p_{ij}A^{r+1}\widehat{\mb{x}}^j_{k-r|k-r} \mb{I}_{k-r,ij}(r),
	\\\label{eq_m}
	\widehat{\mb{x}}^i_{k|k} =& \widehat{\mb{x}}^i_{k|k-1} + K_i C_i^\top \left({y}^i_k-C_i\widehat{\mb{x}}^i_{k|k-1}\right),
\end{align} \normalsize
where  $\mb{I}_{k,ij}(r)$ is the indicator function  defined as \cite{Themis_delay},
\begin{align}
	\mb{I}_{k,ij}(r) = \left\{
	\begin{array}{ll}
		1, & \text{if}~ \tau_{ij}=r  \\
		0, & \text{otherwise}.
	\end{array}\right.
\end{align}
In \eqref{eq_p}, $\widehat{\mb{x}}^i_{k|k-1}$ denotes the sensor $i$'s \textit{a-priori} state estimate at time $k$ given all the (possibly delayed) information up to time $k-1$ from its in-neighbors $\mathcal{N}_i$. Step \eqref{eq_p} represents one iteration of consensus-based information-fusion on all the received information, where sensor $i$ sums the weighted estimates
of sensors $j\in \mc{N}_i$ as they arrive  knowing the delays. The performance analysis of \eqref{eq_p}-\eqref{eq_m} in terms of mean-square stability (for the delay-free case) is given in \cite{jstsp}. Recall that, for every link $(i,j)$ the indicator $\mb{I}_{k-r,ij}(r)$ is only non-zero for one $r$ between $0$ and $\overline{\tau}$ (due to fixed delay assumption). The second step \eqref{eq_m} is a measurement-update (also known as \textit{innovation}) to modify the a-priori estimate based on the new measurement of sensor $i$. Clearly, the  protocol \eqref{eq_p}-\eqref{eq_m} is single time-scale with one step of information-sharing/consensus-update between every $k-1$ and $k$. 
Using the notion of augmented vector, define $\underline{\widehat{\mb{x}}}_{k|k-1} = \left(\widehat{\mb{x}}_{k|k-1}; \widehat{\mb{x}}_{k-1|k-2}; \dots; \widehat{\mb{x}}_{k-\overline{\tau}|k-\overline{\tau}-1} \right)$ and similarly $\underline{\widehat{\mb{x}}}_{k|k}$. Then, the augmented version of \eqref{eq_p}-\eqref{eq_m} is,
\begin{align}\label{eq_p_aug}
	\underline{\widehat{\mb{x}}}_{k|k-1} =&\overline{PA} \underline{\widehat{\mb{x}}}_{k-1|k-1},
	\\\label{eq_m_aug}
	\underline{\widehat{\mb{x}}}_{k|k} =&  \underline{\widehat{\mb{x}}}_{k|k-1} + \mb{b}^{\overline{\tau}+1}_1 \otimes K D_C^\top \left(\mb{y}_k-D_C\Xi^{Nn}_{1,\overline{\tau}}\underline{\widehat{\mb{x}}}_{k|k-1}\right),
\end{align}
where 
$\overline{PA}$ is the \textit{modified} augmented version of $P \otimes A$ as,
\small \begin{align} \label{eq_aug_WA}
	\overline{PA}= \left( 
	\begin{array}{cccccc}
		P_0 \otimes A & P_1\otimes A^2  &  \hdots & P_{\overline{\tau}-1}\otimes A^{\overline{\tau}}  & P_{\overline{\tau}} \otimes A^{\overline{\tau}+1}  \\
		I_{Nn} &   0_{Nn}  &\hdots  & 0_{Nn}& 0_{Nn}\\
		0_{Nn} & I_{Nn} &   \hdots  & 0_{Nn} & 0_{Nn}  \\
		0_{Nn} &  0_{Nn} &  \ddots  & 0_{Nn} & 0_{Nn}  \\
		\vdots & \vdots &  \hdots & \vdots & \vdots \\
		0_{Nn} & 0_{Nn} &  \hdots & I_{Nn} & 0_{Nn}
	\end{array}	
	\right),
\end{align} \normalsize
and
${D}_C=\mbox{blockdiag}(C_i)$, $K = \mbox{blockdiag}(K_i)$, and the  auxiliary matrix $\Xi^m_{i,\overline{\tau}}$ is an $m \times (\overline{\tau}+1)m $ matrix defined as
$\Xi^m_{i,\overline{\tau}}= (\mb{b}^{\overline{\tau}+1}_i \otimes I_m)^\top $
with $\mb{b}^{\overline{\tau}+1}_i$ as the unit column-vector of the $i$'th coordinate ($1\leq i \leq {\overline{\tau}+1}$). 
\vspace{\belowdisplayskip}
\par\noindent\rule{\dimexpr(0.5\textwidth-0.5\columnsep-0.4pt)}{0.4pt}%
\rule{0.4pt}{6pt}
\begin{strip}
Define the augmented state  $\underline{\mb{x}}_{k} = \left(1_{N} \otimes   \mb{x}_k; 1_{N} \otimes \mb{x}_{k-1}; \dots; 1_{N} \otimes \mb{x}_{k-\overline{\tau}} \right) $ and the augmented error $\underline{\mb{e}}_{k}$ at time $k$ as follows,
\begin{align} 
	\underline{\mb{e}}_{k}  &= \underline{\mb{x}}_{k} - \underline{\widehat{\mb{x}}}_{k|k}  =\underline{\mb{x}}_{k} - \Bigl(\underline{\widehat{\mb{x}}}_{k|k-1} + \mb{b}^{\overline{\tau}+1}_1 \otimes K D_C^\top \mb{y}_k-D_C\Xi^{Nn}_{1,\overline{\tau}}\underline{\widehat{\mb{x}}}_{k|k-1}\Bigr) \nonumber
	\\
	&= \underline{\mb{x}}_{k}  - \overline{PA}\underline{\widehat{\mb{x}}}_{k-1|k-1} - \mb{b}^{\overline{\tau}+1}_1 \otimes K D_C^\top \left(\mb{y}_k-D_C\Xi^{Nn}_{1,\overline{\tau}}\overline{PA} \underline{\widehat{\mb{x}}}_{k-1|k-1}\right). \label{eq_love}
\end{align}

Define $\widetilde{\mb{\nu}}_{k} =  1_{N} \otimes \mb{\nu}_{k} $
and $\underline{\mb{\nu}}_{k} =\mb{b}^{\overline{\tau}+1}_i \otimes\widetilde{\mb{\nu}}_{k}$. 
We have $D_C^\top \mb{y} =D_C^\top D_C(1_{N} \otimes\mb{x}_k)+  D_C^\top \mb{\zeta} $. 
Recall that the row-stochasticity of $P$ and $\overline{P}$ along with the system dynamics~\eqref{eq_sys} implies that $\underline{\mb{x}}_{k}   = \overline{PA} \underline{\mb{x}}_{k-1} +  \underline{\mb{\nu}}_{k}  $; substituting this  along with~\eqref{eq_sys}-\eqref{eq_y} into \eqref{eq_love},
\begin{align}
	\underline{\mb{e}}_{k} &= \overline{PA} \underline{\mb{x}}_{k-1} + \underline{\mb{\nu}}_{k} -  \overline{PA}\underline{\widehat{\mb{x}}}_{k-1|k-1} -  \mb{b}^{\overline{\tau}+1}_1 \otimes K D_C^\top \left(D_C(1_{N} \otimes\mb{x}_k)+  \mb{\zeta}_{k}-D_C\Xi^{Nn}_{1,\overline{\tau}}\overline{PA}\underline{\widehat{\mb{x}}}_{k-1|k-1}\right) \nonumber \\
	&= \overline{PA} \underline{\mb{e}}_{k-1} 
	- \mb{b}^{\overline{\tau}+1}_1 \otimes K \overline{D}_C \Bigl(1_{N} \otimes\mb{x}_{k}  
	- \Xi^{Nn}_{1,\overline{\tau}}\overline{PA}\underline{\widehat{\mb{x}}}_{k-1|k-1}\Bigr) + \underline{\mb{\eta}}_k \nonumber  \\
	&= \overline{PA} \underline{\mb{e}}_{k-1} 
	- \mb{b}^{\overline{\tau}+1}_1 \otimes K \overline{D}_C \Bigl(\Xi^{Nn}_{1,\overline{\tau}}\overline{PA}\underline{\mb{x}}_{k-1}  
	- \Xi^{Nn}_{1,\overline{\tau}}\overline{PA}\underline{\widehat{\mb{x}}}_{k-1|k-1}\Bigr) + \underline{\mb{\eta}}_k \nonumber  \\
	&= \overline{PA} \underline{\mb{e}}_{k-1} 
	- \mb{b}^{\overline{\tau}+1}_1 \otimes K \overline{D}_C \Xi^{Nn}_{1,\overline{\tau}}\overline{PA}\underline{\mb{e}}_{k-1}  + \underline{\mb{\eta}}_k   = \underline{\widehat{A}} \underline{\mb{e}}_{k-1}  + \underline{\mb{\eta}}_k, \label{eq_err1}
\end{align}
where $\underline{\widehat{A}}= \overline{PA}
- \mb{b}^{\overline{\tau}+1}_1 \otimes K \overline{D}_C \Xi^{Nn}_{1,\overline{\tau}}\overline{PA}$ is the \textit{closed-loop matrix}, $\overline{D}_C= D_C^\top D_C$, and $\underline{\mb{\eta}}_k$ collects the noise terms as,
\begin{align}
	\underline{\mb{\eta}}_k &= \underline{\mb{\nu}}_{k} - \mb{b}^{\overline{\tau}+1}_1 \otimes K \overline{D}_C\widetilde{\mb{\nu}}_{k} -\mb{b}^{\overline{\tau}+1}_1\otimes  K {D}_C^\top\mb{\zeta}_{k} =  \mb{b}^{\overline{\tau}+1}_1 \otimes (\widetilde{\mb{\nu}}_{k} - K \overline{D}_C\widetilde{\mb{\nu}}_{k} - K {D}_C^\top\mb{\zeta}_{k}).
	\label{eq_eta}
\end{align}
\end{strip}
\hfill\rule[-6pt]{0.4pt}{6.4pt}%
\rule{\dimexpr(0.5\textwidth-0.5\columnsep-1pt)}{0.4pt}

The error dynamics in the absence of any delay is as follows,
\begin{align} \nonumber
	\mb{e}_{k} &= (P\otimes A - K \overline{D}_C (P\otimes A))\mb{e}_{k-1} +
	\mb{\eta}_k \\
	&= \widehat{A} \mb{e}_{k-1} + 	\mb{\eta}_k 
	\label{eq_err_nodelay}
\end{align}
where  $\mb{\eta}_k$ follows  the formulation \eqref{eq_eta} with $\overline{\tau}=0$ and $\widehat{A}=P\otimes A - K \overline{D}_C (P\otimes A)$ is the \textit{ delay-free closed-loop matrix}. 
For the Schur stability of the error dynamics \eqref{eq_err1} and \eqref{eq_err_nodelay}, we need $\rho(\underline{\widehat{A}})<1$ and $\rho(\widehat{A})<1$, respectively. We first discuss the condition for Schur stability of $\widehat{A}$ and then extend the results to Schur stability of $\underline{\widehat{A}}$. Following Kalman theorem  and  justification in \cite{jstsp,usman_cdc:11}, for stability of \eqref{eq_err_nodelay} the pair $(P\otimes A, \overline{D}_C)$ needs to be observable (or detectable); this is known as \textit{distributed observability} \cite{jstsp}, discussed next. 

\begin{lem} \label{lem_irreducible}
	Given a full-rank matrix $A$ and output matrix $C$, following Assumption~\ref{ass_AC} and Remark~\ref{rem_ACj}, the pair $(P\otimes A, \overline{D}_C) $ is (structurally) observable if the matrix $P$ is irreducible.
\end{lem}
\begin{proof}
	The proof follows the results in \cite{tsipn_kron} on the (structural) observability of  composite Kronecker-product networks. Given a system digraph $\mc{G}_1$ associated with full-rank system $A$ and measurement matrix $C$ satisfying $(A,C)$-observability, the minimum sufficient condition for observability of the Kronecker-product network (denoted by $\mc{G} \times \mc{G}_1$) is that $\mc{G}$ be strongly-connected and self-damped (see Theorem 4 in \cite{tsipn_kron}). Following the definition of consensus matrix $P$, we have $p_{ii} \neq 0$ (satisfying the self-damped condition). The strong-connectivity of the  sensor network $\mc{G}$ is equivalent with irreducibility of matrix $P$, which completes the proof. 
\end{proof}

\begin{cor}\label{cor_rho}
For observable $(P\otimes A, \overline{D}_C) $, the gain matrix $K$ can be designed such that  $\rho(\widehat{A})<1$. 
\end{cor}

\subsection{Constrained Feedback Gain Design}
It is known that for an observable pair $(P \otimes A, \overline{D}_C)$, the feedback gain ${K}$ can be designed to ensure Schur stability of the error dynamics~\eqref{eq_err_nodelay} (Corollary~\ref{cor_rho}), i.e., $\rho(\widehat{A})<1$. Typically, such $K$ is designed via solving the following LMI,
\begin{eqnarray} \label{eq_succ}
	X - \widehat{A}^\top X  \widehat{A} \succ 0,
\end{eqnarray}
for some $X\succ 0 $ with $\succ$ implying positive-definiteness. The solution of \eqref{eq_succ} is, in general, a \textit{full} matrix. However, for distributed estimation, we need the state feedback to be further \textit{localized}, i.e., the gain matrix ${K}$ needs to be \textit{block-diagonal} so every sensor uses its own state-feedback.  Such a constrained feedback gain design is proposed in \cite{usman_cdc:11,jstsp} based on cone-complementarity LMI algorithms, which are known to be of polynomial-order complexity for application in large scale. 


\subsection{Stability of the Delayed Estimator Dynamics}
Following the Schur stability of the delay-free error dynamics \eqref{eq_err_nodelay}  (via LMI design of  $K$), we extend the results to stability of the delayed dynamics \eqref{eq_err1}, i.e., to get  $\rho(\underline{\widehat{A}})<1$ in the presence of  delays.
\begin{thm} \label{thm_tau*}
	Let conditions in Lemma~\ref{lem_irreducible} hold and the feedback gain $K$ is designed such that $\rho({\widehat{A}})<1$ from Corollary~\ref{cor_rho}. The networked estimator \eqref{eq_p}-\eqref{eq_m}  successfully tracks the system \eqref{eq_sys} (subject to delays satisfying Assumption \ref{ass_tau} and possibly with $\rho({A})>1$) with stable error for any $\overline{\tau}\leq \overline{\tau}^*$, where $\overline{\tau}^* = \argmax_{\overline{\tau}} 
	\{\rho(P\otimes A^{\overline{\tau}+1} - K \overline{D}_C (P\otimes A^{\overline{\tau}+1}) )< 1\}$.
\end{thm}
\begin{proof}
	For the proof, following from Lemma~\ref{lem_eig_aug} (in the Appendix), we show that $\rho(\underline{\widehat{A}})<1$ for $\overline{\tau}\leq \overline{\tau}^*$, implying Schur stable error dynamics \eqref{eq_err1}. Recall that $\rho({\widehat{A}})<1$ implies that the networked estimator \eqref{eq_p}-\eqref{eq_m}  successfully tracks  system  \eqref{eq_sys} for the delay-free case. From \eqref{eq_aug_WA} and Lemma~\ref{lem_eig_aug}, for the closed-loop matrix $\underline{\widehat{A}}$ (as modified augmented version of ${\widehat{A}}$) and $\rho({A})>1$, 
	\begin{align} \label{eq_tau*}
	    \rho(\underline{\widehat{A}})\leq \rho(P\otimes A^{\overline{\tau}+1} - K \overline{D}_C (P\otimes A^{\overline{\tau}+1}) )^{\frac{1}{\overline{\tau}+1}},
	\end{align}
	which implies that $\rho(\underline{\widehat{A}})<1$ for any $\overline{\tau}\leq \overline{\tau}^*$. In case $\rho({A})<1$, since  $\rho({A})^{\overline{\tau}+1}<\rho({A})$, Schur stability of $\widehat{A}$ also ensures the stability of $\underline{\widehat{A}}$ (for all $\overline{\tau}$). This completes the proof.
\end{proof}	
This theorem gives a sufficient condition for  stable tracking in the presence of  heterogeneous, time-invariant delays $\tau_{ij} = \overline{\tau}$.

\subsection{Convergence Rate}
Note that, in general, the exact characterization of the convergence rate/time of the linear systems is difficult. The following lemma gives the order of convergence time. 
\begin{lem}  \label{lem_rate}
	The convergence time of the error dynamics \eqref{eq_err1} and \eqref{eq_err_nodelay} are  of order\footnote{Given functions $f(\cdot)$ and $g(\cdot)$, say $f = O(g)$ if
$\sup_n |\frac{f(n)}{g(n)}| < \infty$ and $f = \Omega(g)$ if $g = O(f)$. If both $f = O(g)$ and
$f = \Omega(g)$ holds, then $f = \Theta(g)$ \cite{ghaderi2014opinion}.}  $\Theta(\frac{1}{1-\rho(\underline{\widehat{A}})})$ and $\Theta(\frac{1}{1-\rho(\widehat{A})})$, respectively. 
\end{lem}
\begin{proof}
	The proof follows from Lemma 3 in \cite{ghaderi2014opinion}.
\end{proof}
 Following Lemma~\ref{lem_eig_aug} and \ref{lem_rate}, the geometric decay rate of  \eqref{eq_err1} (for $\tau_{ij}=\overline{\tau}$) is  proportional to, 
 \begin{align} \label{eq_rate}
     1-\rho(P\otimes A^{\overline{\tau}+1} - K \overline{D}_C (P\otimes A^{\overline{\tau}+1}) )^{\frac{1}{\overline{\tau}+1}}
 \end{align}
For longer delays (i.e., greater $\overline{\tau},\overline{\tau}^*$ while $\overline{\tau}\leq \overline{\tau}^*$), the consensus rate in \eqref{eq_p}  and, in turn,  the decay rate of overall error \eqref{eq_err1} is slower. In fact, the convergence rate  is lower-bounded by \eqref{eq_rate} (for time-invariant delays).
For linear feedback systems, one can easily adjust the closed-loop eigenvalues and the convergence rate by  design of the desired feedback gain $K$. However, in the decentralized case,  $K$ is constrained to be block-diagonal. Such LMI-design with additional bound-constraint on closed-loop eigenvalues is a complex problem.   
  
\subsection{Discussions}
\begin{enumerate}
	\item Each sensor processes the a-priori estimates of its in-neighbors as they arrive. The messages are time-stamped and the sensor knows the time-step (and hence the delay) of the received information. The  proposed solution works \textit{for both heterogeneous and homogeneous delays}. 
	\item The observability results are independent of the specific choice of the consensus weights (e.g., lazy Metropolis \cite{bu2018accelerated} or simply $\frac{1}{|\mc{N}_i|+1}$ \cite{Themis_delay})  to satisfy stochasticity of $P$. This is due to  generic/structural observability results which hold for almost all choices of numerical entries of $P$ as long as its structure (the sensor network) is fixed/time-invariant \cite{commault2018classification,icassp2016}. Recall from structural observability that the weights for which the system is unobservable are of zero Lebesgue measure, i.e., if we choose the weights randomly, then the system will be almost surely observable. However, the structure of the consensus fusion matrix $P$ may affect  LMI-based gain design, the bound in Eq.~\eqref{eq_tau*}$, \rho(\underline{\widehat{A}})$, and the convergence rate \eqref{eq_rate}. Note from Lemma~\ref{lem_eig_kron} (in the Appendix) that $\rho(P \otimes A)= \rho(A)$ since $\rho(P)=1$ for any choice of row-stochastic $P$.  
	
	\item The cost-optimal design  of the sensor network structure and sensor placement \cite{spl18,pequito2014optimal} can be considered to reduce the communication-related and/or sensing-related costs. In general, the cost-optimal design  subject to strong-connectivity is NP-hard. However, considering bidirectional links among sensors, it has a solution with polynomial-order complexity $\mc{O}(N^2)$ \cite{spl18}. 
	\item In case of sensor failure, the concept of \textit{observational equivalence} in both \textit{centralized} \cite{commault2018classification} and \textit{distributed} \cite{icassp2016} scenarios can recover the loss of observability.
	\item To design a distributed estimator to tolerate time-delays bounded by $\overline{\tau}_1$, one can redesign the LMI gain matrix $K$ by replacing $\widehat{A}= P\otimes A^{\overline{\tau}_1+1} - K \overline{D}_C (P\otimes A^{\overline{\tau}_1+1})$ in \eqref{eq_succ}. Clearly, from Theorem~\ref{thm_tau*} and \eqref{eq_tau*}, such $K$ results in $\rho(\underline{\widehat{A}})<1$ for $\overline{\tau}\leq \overline{\tau}_1$ (simply replace $\overline{\tau}^*=\overline{\tau}_1$). Such LMI gain design for the delay-free closed-loop matrix $\widehat{A}$ of size $nN$ instead of the delayed  matrix $\underline{\widehat{A}}$ of size $nN(\overline{\tau}+1)$, significantly reduces the  complexity order  with no need of using the augmented matrix $\overline{P}$.
	\item In \cite[Theorem 1]{HOFBAUER},  it is claimed that if a consensus matrix $P$ is \textit{weakly diagonally dominant}, then the off-diagonal delays in $P$ are harmless for stability. However, in error dynamics \eqref{eq_err1}, the entries of $\underline{\widehat{A}}$ (and the weak/strong diagonal dominance of the closed-loop system) depends on the feedback gain $K$ and cannot be evaluated only based on the open-loop matrices $A$ and $P$. 
\end{enumerate}
%

 \section{Simulation} \label{sec_sim}
For MATLAB simulation, we consider a linear structurally-cyclic system of $n=6$ states with $4$ irreducible sub-systems and a group of $N=4$ sensors each taking one system output (from each irreducible block) satisfying Assumption~\ref{ass_AC}. The system is full-rank and unstable with $\rho(A)=1.04$. System and output noise are considered as $\mc{N}(0,0.004)$. The network of sensors $\mc{G}$ is considered as a simple directed self-damped cycle $1 \rightarrow 2 \rightarrow 3 \rightarrow 4 \rightarrow 1$, where  the system is not observable in the neighborhood of any sensor  (Remark~\ref{rem_ACj}). With this structure, the consensus weights are considered random while satisfying row-stochasticity of $P$. It is clear that such irreducible $P$ matrix satisfies Lemma~\ref{lem_irreducible} and, therefore, using the LMI strategy in \cite{usman_cdc:11,jstsp}, the block-diagonal gain matrix $K$ is designed such that $\rho(\widehat{A})=0.64<1$ and $\overline{\tau}^* =10$, implying stable error dynamics for any $\overline{\tau} \leq 10$ (sufficiency from Theorem~\ref{thm_tau*}). Next, considering both  heterogeneous delays (\textit{uniformly distributed} between $0$ and $\overline{\tau}$ for different links \cite{Themis_delay}) and  homogeneous delays (equal to $\overline{\tau}$ at all links) following Assumption~\ref{ass_tau},  the performance of the distributed estimation is analyzed. Fig.~\ref{fig_delay_mc} shows the Monte-Carlo simulation ($100$ trials) of mean-squared error (MSE) over the network with, (i) no time-delay, (ii) homogeneous/heterogeneous delays with $\overline{\tau}=3,8,19$  and (sufficiency) bounds in \eqref{eq_tau*} as  ${\rho(\underline{\widehat{A}})\leq 0.74, 0.94, 1.45}$. For fixed homogeneous delays $\overline{\tau}=19$, we have ${\rho(\underline{\widehat{A}})=1}$. From Fig.~\ref{fig_delay_mc} and Eq. \eqref{eq_rate}, longer delays decrease the MSE decay rate for $\overline{\tau}\leq \overline{\tau}^*$ (with $\overline{\tau}^*=10$ for this example), while for $\overline{\tau}> \overline{\tau}^*$ the error  \textit{may not necessarily} converge.

\begin{figure} 
	\centering
	\includegraphics[width=3.25in]{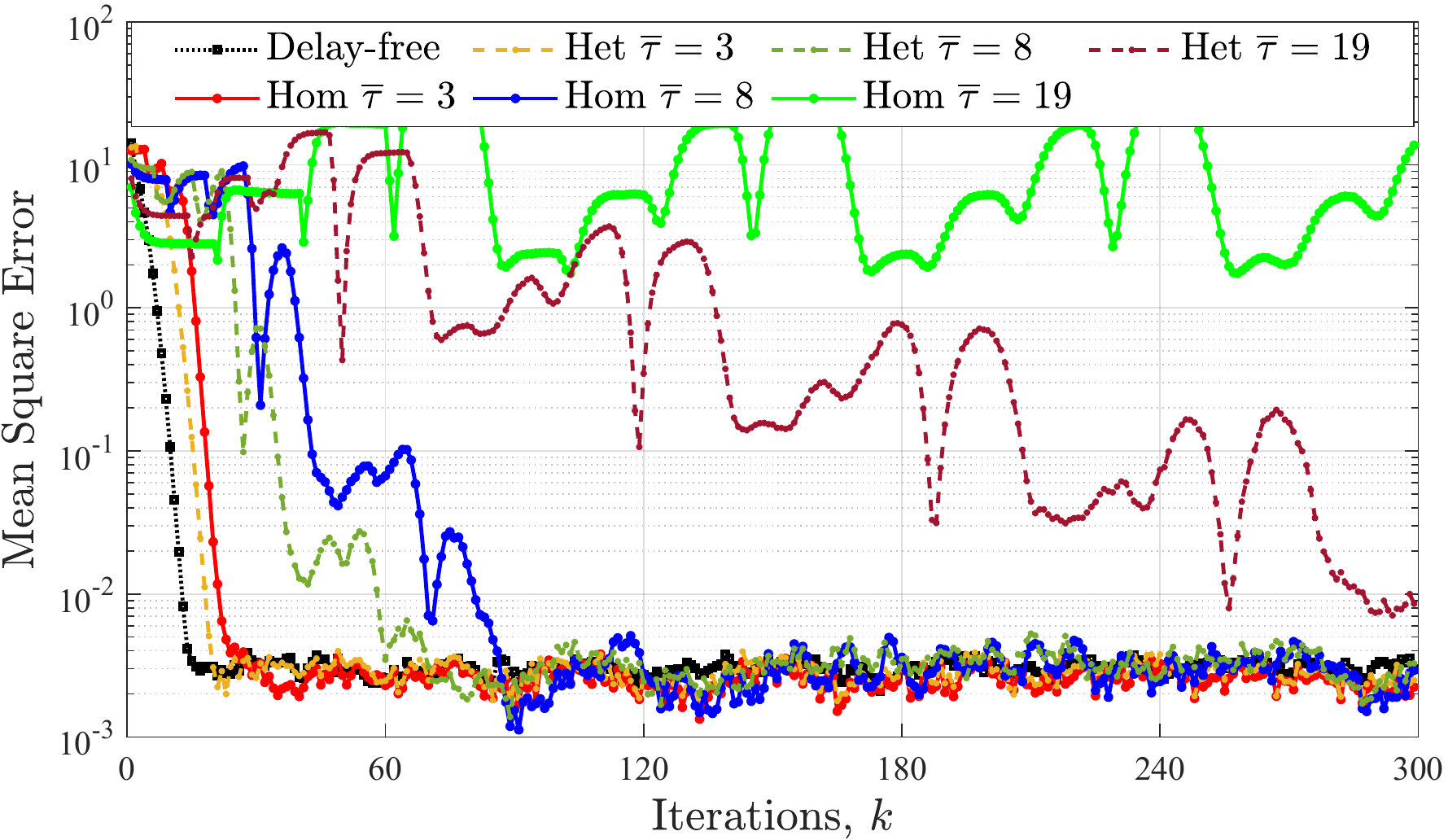}
	\caption{MSEs are bounded steady-state stable for $\overline{\tau}=0,3,8$. For $\overline{\tau}=19> \overline{\tau}^*=10$ the MSE is marginally unstable for the homogeneous case, while it is stable for the heterogeneous case. Note that  Theorem~\ref{thm_tau*} gives the \textit{sufficiency condition (on $\overline{\tau}$) for convergence not necessity}.
	} \label{fig_delay_mc}
\end{figure}

\section{Conclusions and Future Directions} \label{sec_conc}

This paper extends the recent literature on distributed estimation over  linear networks to time-delayed ones. We assume heterogeneous time-invariant  delays for the communication links. For a given bound on the delays, we provide a solution to design distributed estimators enabling all sensors to successfully track the system state over delayed networks. 

Part of ongoing research focuses on rank-deficient systems, which adds more complexity to the problem in terms of system outputs, network connectivity, and data-sharing  \cite{icassp13,pequito2015framework}. 
Other promising research directions are (i) detecting sensor faults/attacks \cite{deghat2019detection,tcns_fdi} along with considering latency on the distributed estimation networks, (ii) extension to time-varying delays as in \cite{Themis_delay}, and (iii)
pruning the network to improve the observability properties and convergence \cite{lcss2020,pirani_eigvalue}.


\section*{Appendix}
Some of the following lemmas can be found in standard matrix theory books, e.g., in \cite{bhatia2013matrix}.
\begin{lem} \label{lem_eig_kron}
	Consider two square matrices $P$ and $A$ of size $N$ and $n$, respectively, with the set of eigenvalues $\{\lambda_1,\hdots,\lambda_N\}$ and $\{\mu_1,\hdots,\mu_n\}$. Then, the set of eigenvalues of $P \otimes A$ is $\{\lambda_i \mu_j| i=1,\hdots,N,~j=1,\hdots,n \}$.
\end{lem}
\begin{lem} \label{lem_polynom}
    Define the following $nN$-by-${nN}$ block matrix,
\small	\begin{align}\label{eq_aug_A}
		 \overline{A}_{n,i}= \left( 
		\begin{array}{ccccc}
			0_N & \hdots & A_{i} & \hdots & 0_N \\
			I_N &   0_N & \hdots &\hdots & 0_N\\
			0_N &  I_N & \ddots   &  \hdots   & 0_N  \\
			\vdots & \vdots & \ddots &  \ddots & \vdots \\
			0_N & 0_N &  \hdots & I_N & 0_N
		\end{array}	
		\right),
	\end{align} \normalsize
	where $N$-by-${N}$ matrix $A_{i}$ is located at the $i$th block (and the only non-zero block) in the first block-row of $\overline{A}_{n,i}$. Let $p(\lambda)$ and $q(\overline{\lambda})$  represent the characteristic polynomials of $A_{i}$ and $\overline{A}_{n,i}$, respectively. Then, $q(\overline{\lambda}) = \overline{\lambda}^{N(n-i)}p(\overline{\lambda}^i)$. 
\end{lem}
\begin{proof}
Consider,
\begin{align}
 		 \overline{\lambda} I_{nN} - \overline{A}_{n,i}= 
		 \left(\begin{array}{cc}
			E  & F\\
			G & H
		\end{array}\right),	
\end{align}
where block-matrix $E$ is $N(i-1)$-by-$N(i-1)$,  $F$ is $N(i-1)$-by-$N(n-i+1)$,  $G$ is $N(n-i+1)$-by-$N(i-1)$, and $H$ is $N(n-i+1)$-by-$N(n-i+1)$ defined as,	
\footnotesize
\begin{align}  
		 E &= \left( 
		\begin{array}{ccccc}
			\overline{\lambda} I_N & 0_N & \hdots & \hdots &  0_N \\
			-I_N &  \overline{\lambda} I_N  & \hdots &\hdots & 0_N\\
			0_N &  -I_N & \ddots   &  \hdots   & 0_N  \\
			\vdots & \vdots & \ddots &  \ddots & \vdots \\
			0_N & 0_N &  \hdots & -I_N & \overline{\lambda} I_N
		\end{array}	
		\right)\\
        F &= \left( 
		\begin{array}{cccc}
			-A_i & 0_N & \hdots  & 0_N \\
			0_N &   0_N & \hdots & 0_N\\
			\vdots & \vdots & \vdots &   \vdots \\
			0_N & 0_N &  \hdots  & 0_N
		\end{array}	
		\right) G =\left( 
		\begin{array}{cccc}
			0_N &  \hdots & 0_N & -I_N   \\
			  0_N & \hdots & 0_N & 0_N\\
			\vdots & \vdots & \vdots & \vdots \\
			0_N & 0_N &  \hdots &  0_N
		\end{array}
		\right) \\
		H &= \left( 
		\begin{array}{ccccc}
			\overline{\lambda} I_N & 0_N & \hdots & \hdots & 0_N \\
			-I_N &  \overline{\lambda} I_N & \hdots &\hdots & 0_N\\
			0_N &  -I_N & \ddots   &  \hdots   & 0_N  \\
			\vdots & \vdots & \ddots &  \ddots & \vdots \\
			0_N & 0_N &  \hdots & -I_N & \overline{\lambda} I_N
		\end{array}
		\right).
	\end{align} \normalsize
Recall that $p(\lambda)= |\lambda I_N - A_i|$, and 
\begin{align} \label{eq_gfeh}
		q(\overline{\lambda}) = |\overline{\lambda} I_{nN} - \overline{A}_{n,i}|= |E||H-GE^{-1}F|.
\end{align}
We have,
\small		\begin{align}
		 E^{-1} &=& \left( 
		\begin{array}{ccccc}
			\frac{I_N}{\overline{\lambda}} & 0_N & \hdots & \hdots & 0_N \\
			\frac{I_N}{\overline{\lambda}^2} &   \frac{I_N}{\overline{\lambda}} & \hdots &\hdots & 0_N\\
			\vdots & \vdots & \ddots &  \ddots & \vdots \\
			\frac{I_N}{\overline{\lambda}^{i-1}} & \frac{I_N}{\overline{\lambda}^{i-2}} &  \hdots & \frac{I_N}{\overline{\lambda}^{2}} & \frac{I_N}{\overline{\lambda}}
		\end{array}	
		\right), 
\end{align} \normalsize
and $H-GE^{-1}F$ is equal to,

\small \begin{align}
 \left( 
		\begin{array}{ccccc}
			\overline{\lambda} I_N-\frac{A_i}{\overline{\lambda}^{i-1}} & 0_N & \hdots & \hdots & 0_N \\
			-I_N &  \overline{\lambda} I_N & \hdots &\hdots & 0_N\\
			0_N &  -I_N & \ddots   &  \hdots   & 0_N  \\
			\vdots & \vdots & \ddots &  \ddots & \vdots \\
			0_N & 0_N &  \hdots & -I_N & \overline{\lambda} I_N
		\end{array}
		\right).
	\end{align} \normalsize
Then, $|E|= \overline{\lambda}^{N(i-1)}$, $ |H-GE^{-1}F|=\overline{\lambda}^{N(n-i)}|\overline{\lambda} I_N-\frac{A_i}{\overline{\lambda}^{i-1}}|$, and substituting these in \eqref{eq_gfeh},
\begin{align} 
		q(\overline{\lambda}) =  \overline{\lambda}^{N(n-i)}|\overline{\lambda}^i I_N - A_i| = \overline{\lambda}^{N(n-i)} p(\overline{\lambda}^i).
\end{align}
The proof is complete.
\end{proof}


\begin{lem} \label{lem_eig_aug}
    Given matrix $A$ with $\rho(A)<1$, we have $\rho(\overline{A})\leq \rho(A)^{\frac{1}{\overline{\tau}+1}}<1$ with  $\overline{A}$ as the augmented form of $A$ via Eq. \eqref{eq_aug_W}.
\end{lem}
\begin{proof}
	The characteristic polynomial of $\overline{A}$ can be defined based on Lemma~\ref{lem_polynom}. Let  $p(\lambda)$ and $q(\overline{\lambda})$ respectively  represent the characteristic polynomial of $A$ and $\overline{A}$. 
For $\tau_{ij} = \overline{\tau}$ for all $i,j$ and $A_{\overline{\tau}} = A$. Therefore, $ q(\overline{\lambda}) = p(\overline{\lambda}^{\overline{\tau}+1})$ and $	\rho(\overline{A}) = \rho(A)^{\frac{1}{\overline{\tau}+1}}<1.$
We know that the function $\rho(A)^{\frac{1}{{\tau}+1}}$ is an increasing function of $\tau$ (given $\rho(A)<1$), then, 
for $\tau_{ij} =r<\overline{\tau}$ for all $i,j$ and $A_{r} = A$, we have $\rho(\overline{A}) = \rho(A)^{\frac{1}{r+1}}< \rho(A)^{\frac{1}{\overline{\tau}+1}} <1.$
This can be generalized for any choice of bounded time-delay and associated augmented matrix 
in the form of \eqref{eq_aug_W}. This completes the proof.
\end{proof}

\bibliographystyle{IEEEbib}
\bibliography{bibliography}
\end{document}